\documentclass[smallcondensed]{svjour3}     % onecolumn (ditto)
\smartqed  % flush right qed marks, e.g. at end of proof
\usepackage{graphicx}
\usepackage[colorlinks=false]{hyperref}
\usepackage{float}
\usepackage{array}
\usepackage{epsfig}
\usepackage{color}
\usepackage{rotating}
\usepackage{amssymb}
\usepackage{amsmath,amsfonts}
\usepackage[latin1]{inputenc}
\usepackage{verbatim}
\usepackage[tableposition=top,font=small,labelfont=bf,format=hang]{caption}
\usepackage{booktabs}
\usepackage{enumerate}
\usepackage{arydshln}
\newtheorem{lem}{Lemma}

\newtheorem{cor}{Corollary}

\newcommand{\G}{\mbox{$\mathcal G$}}

\newcommand{\D}{\mbox{$\mathcal D$}}

\newcommand{\C}{\mbox{$\cal C$}}
\newcommand{\Tr}{\mbox{Tr}}
\newcommand{\ob}{\mbox{$\overline{\omega}$}}
\newcommand{\om}{\mbox{$\omega$}}
\newcommand{\F}{\mathbb{F}}

\DeclareMathOperator{\wt}{wt}
\begin{document}

\title{Additive complementary dual codes over $\F_4$\thanks{This research is supported by National Natural Science Foundation of China (12071001).}}
\author{ Minjia Shi         \and
        Na Liu         \and
        Jon-Lark Kim       \and
        Patrick Sol\'e
}
\institute{M. Shi \at
              Key Laboratory of Intelligent Computing Signal Processing, Ministry of Education, School of Mathematical Sciences, Anhui University, Hefei, Anhui, 230601, China. \\
              \email{smjwcl.good@163.com}
           \and
           Na Liu  \at
              School of Mathematical Sciences, Anhui University, Hefei, Anhui, 230601, China.\\
              \email{naliu1177@163.com}
           \and
           Jon-Lark Kim \at
          Department of Mathematics, Sogang University, Seoul, South Korea\\
           \email{jlkim@sogang.ac.kr}
           \and
           P. Sol\'e \at
           Aix Marseille Univ, CNRS, Centrale Marseille, I2M, Marseille, France\\
           \email{sole@enst.fr}
}

\date{Received: date / Accepted: date}

% \author{ Minjia Shi\thanks{smjwcl.good@163.com},
 % Na Liu \thanks{naliu1177@163.com}, Jon-Lark Kim\thanks{jlkim@sogang.ac.kr},
%Patrick Sol\'e\thanks{{sole@enst.fr}}
%\thanks{Minjia Shi and Na Liu are with School of Mathematical Sciences, Anhui University, Hefei, 230601, China. Jon-Lark Kim is with Department of Mathematics, Sogang University, Seoul, South Korea. Patrick Sol\'e is with Aix Marseille Univ, CNRS, Centrale Marseille, I2M, Marseille, France. }}

%\author{Minjia Shi, Shitao Li, Jon-Lark Kim\thanks{Department of Mathematics,
%Sogang University,
%Seoul 04107, South Korea.
%{Email: \tt jlkim@sogang.ac.kr},},  Patrick Sole
%J.-L. Kim was supported by Basic Research Program through the National Research Foundation of Korea (NRF) funded by the Ministry of Education (NRF-2016R1D1A1B03933259).
%}

\date{}

\maketitle

\begin{abstract}
A linear code is linear complementary dual (LCD) if it meets its dual trivially.
LCD codes have been a hot topic recently due to Boolean masking application in the security of embarked electronics (Carlet and Guilley, 2014).
Additive codes over $\F_4$ are $\F_4$-codes that are stable by codeword addition but not necessarily by scalar multiplication. An additive code over $\F_4$ is additive complementary dual (ACD) if it meets its dual trivially.
The aim of this research is to study such codes which meet their dual trivially.
All the techniques and problems used to study LCD codes  are potentially relevant to ACD codes.
Interesting constructions of ACD codes from binary codes are given with respect to the trace Hermitian and trace Euclidean inner product. The former product is relevant to
quantum codes.
\end{abstract}\\



\section{Introduction}
To begin with, let us recall some basic definitions.
A {\em linear $[n,k]$ code} over a finite field $\mathbb F_q$ is a $k$-dimensional subspace of $\mathbb F_q^n$.
The set of vectors orthogonal to $\C$ under the usual inner product is the {\em dual} of $\C$, denoted by $\C^{\perp}$.
A linear code $\C$ is {\em self-orthogonal} if $\C \subset \C^{\perp}$, and {\em self-dual} if $\C = \C^{\perp}$. On the other hand, a linear
code $\C$ is an {\em LCD code} (linear complementary dual code) if $\C \cap \C^{\perp}= \{0\}$. In particular, $\C$ is a {\em binary LCD code} if $\C$ is a binary linear code satisfying $\C \cap \C^{\perp}= \{0\}$.

Massey~\cite{Massey} introduced the notion of LCD codes in order to provide an optimum linear coding
solution for the two-user binary adder channel.
In a later work~\cite{Massey92}, he also showed that there exist asymptotically good LCD codes. Furthermore, Sendrier showed that LCD codes meet
the asymptotic Gilbert-Varshamov bound in \cite{Sen}.

In 2014, Carlet and Guilley~\cite{CarGui} introduced several constructions of LCD codes and investigated an application of LCD codes against
Side-Channel Attacks (SCA) and Fault Injection Attacks (FIA).
Recall from~\cite{CarGui} that SCA consists in passively recording some leakage that is the source of information to retrieve the key,
and that FIA consist in actively perturbating the device so as to obtain exploitable differences at the output.

In the approach of \cite{CarGui}, the direct sum $\C \oplus \C^\bot=\F_q^n$ is essential, and the minimum distance of $\C$ (resp. $\C^\bot$) acts
as a performance criterion  for SCA (resp. FIA).  Since this model does not use the linearity of $\C$ but only its additivity, it makes sense to study additive complementary dual (ACD)
codes over finite fields or finite rings. Moreover, since linear codes are additive codes, LCD codes are ACD. Furthermore, Guilley~\cite{Gui} reported to
us that the application of ACD codes to security still makes sense. This motivates the current study.

In the same spirit, Shi et al.~\cite{ShiLiKimSol} studied ACD codes over a noncommutative non-unital ring $E$ with four elements recently. Nevertheless,
little is known about a general theory of ACD codes over $\mathbb F_4$.
All the techniques and problems used to study LCD codes \cite{D+} are potentially relevant to ACD codes.
Interesting constructions of ACD codes from binary codes are given with respect to two trace inner products: the trace Euclidean inner product, and the trace Hermition inner product, familiar since the studies of quantum codes
\cite{CRSS}. We have also constructed ACD $(6, 2^{5}, 4)$, $(35, 2^7, 26)$, and $(96, 2^7, 72)$ codes over $\F_4$ under the trace Euclidean inner product, all of which have more codewords than optimal LCD $[6,2,4]$, $[35,3,26]$, and $[96,3,72]$ codes over $\F_4$, respectively.

Our paper consists of five sections. Section 2 recalls basic definitions and notations from additive codes under the two inner products. Section 3
discusses ACD codes with respect to the trace Hermitian inner product. Section 4 discusses ACD codes with with respect to the trace Euclidean inner product. Section 5
concludes the article.

%In a first pass the student can look for conditions bearing on the subfield subcode and the trace code for the code to be LCD.

%\section{Problems for ACD codes over $\F_4$}

%We plan to solve some of the following problems.

%\begin{enumerate}

%\item Find a criteria on the generator matrix for an additive $\F_4$-code to be an additive complementary dual(ACD) code.

%\item Construct or characterize additive cyclic $\F_4$-codes.

%\item Construct ACD codes over $\F_4$ from LCD codes over $\F_2$.

%\item Make a table for the upper bound of minimum distances of ACD codes over $\F_4$. Give a linear programming bound.

%\item When is a symmetric $(0,1)$-matrix $A$ of order $k$ with diagonals zero invertible? In other words, let $A$ be an adjacency matrix of a graph without a loop. When is $A$ invertible over $\F_2$. In particular, if the graph is a SRG, when is its adjacency matrix $A$ is invertible over $\F_2$?
%\end{enumerate}

\section{Preliminaries}
Let $\F_2=\{0,1\}$ and $\F_4=\{0,1,\omega,\bar{\omega}\}$ denote the finite field of order 2 and 4 respectively, where $\bar{\omega}=\omega^2=\omega+1$, $\omega^{3}=1$.

An {\em additive code \C\ over $\F_4$ of length $n$}
is an additive subgroup of $\F_4^n$.
As \C\ is a free $\F_2$-module, it has
size $2^k$ for some $0\le k\le 2n$.  We call \C\ an $(n,2^k)$ code.
It has a basis, as a $\F_2$-module, consisting of $k$ basis vectors. Interest in additive codes over $\F_4$ has arisen because of their
correspondence to quantum codes as described in \cite{CRSS}.  There
is a natural inner product on the additive codes
arising from the trace map. The
{\em trace} map $\Tr:\F_4\rightarrow \F_2$ is given by
\[\Tr(x)=x+x^2.\]
In particular $\Tr(0)=\Tr(1)=0$ and $\Tr(\omega)=\Tr(\ob)=1$.
The {\em conjugate} of $x\in \F_4$, denoted $\overline{x}$, is the image
of $x$ under the Frobenius automorphism; in other words, $\overline{0}=0$,
$\overline{1}=1$, and $\overline{\ob}=\omega$.

\begin{definition}
A {\em generator matrix of an $(n,2^k)$ additive code $\C$ over $\F_4$} is a $k\times n$ matrix $G$ with entries in $\F_4$ such that $\C=\{{\bf u}G:\ {\bf u}\in \F_2^k\}$. Note that $G$ has 2-rank $k$.
\end{definition}

As usual, the {\em weight} $\wt({\bf c})$ of ${\bf c}\in\C$ is the number of
nonzero components of ${\bf c}$.  The minimum weight $d$ of \C\ is the smallest
weight of any nonzero codeword in \C.  If \C\ is an $(n,2^k)$ additive code
of minimum weight $d$, \C\ is called an $(n,2^k,d)$ code.

\begin{example}\label{ex:hexacode_quan}
 Let $\G_6$ be the $[6,3,4]$ {\em hexacode} whose generator matrix as a
linear $\F_4$-code is
\[{\setlength{\arraycolsep}{1.785pt}
\renewcommand{\arraystretch}{.7}
\left[\begin{array}{cccccc}
1&0&0&1&\omega&\omega\\
0&1&0&\omega&1&\omega\\
0&0&1&\omega&\omega&1
\end{array}\right].
}\]
This is also an additive code; thinking of $\G_6$ as an additive code,
it has generator matrix
\[{\setlength{\arraycolsep}{1.785pt}
\renewcommand{\arraystretch}{.3}
 \left[\begin{array}{cccccc}
1&0&0&1&\omega&\omega\\
0&1&0&\omega&1&\omega\\
0&0&1&\omega&\omega&1\\
\omega&0&0&\omega&\ob&\ob\\
0&\omega&0&\ob&\omega&\ob\\
0&0&\omega&\ob&\ob&\omega
\end{array}\right].
}\]
\end{example}

%{\bf Question: Represent this additive generator matrix in the form of Definition 1.}

%\medskip

For $\textbf x=(x_1,x_2,\ldots,x_n)$ and $\textbf y=(y_1,y_2,\ldots,y_n)$ in $\F_4^n$, we define the Hermitian inner product,
the trace Hermitian inner product and the trace Euclidean inner product of $\textbf x$ and $\textbf y$ as follows:
$$
(i)~ \textbf x\cdot \textbf y=\sum_{i=1}^n x_i\bar{y_i},\
(ii)~ \textbf x\star \textbf y=\sum_{i=1}^n{\rm Tr}(x_i\bar{y_i}),\
(iii)~ \textbf x\diamond \textbf y=\sum_{i=1}^n{\rm Tr}(x_iy_i).$$

If $\C$ is an $(n,2^k)$ additive code, the dual of $\C$ with respect to the Hermitian inner product, the trace Hermitian inner product and the trace Euclidean inner product are defined as follows:
$$  \C^{\perp_{\rm H}}=\{u\in \F_4^n:u\cdot v=0 \ for \ all \ v\in \C\},$$

$$  \C^{\perp_{{\rm TrH}}}=\{u\in \F_4^n:u\star v=0 \ for \ all \ v\in \C\},$$

$$  \C^{\perp_{{\rm TrE}}}=\{u\in \F_4^n:u \diamond v=0 \ for \ all \ v\in \C\}.$$

Obviously, $\C^{\perp_{\rm H}}$, $\C^{\perp_{{\rm TrH}}}$ and $\C^{\perp_{{\rm TrE}}}$ are $(n, 2^{2n-k})$ additive codes. If, in addition, $\C$ is linear, we say $\C$ is {\em linear complementary dual with respect to the Hermitian inner product (or quaternary Hermitian LCD)} if $\C \cap \C^{\perp_{{\rm H}}}=\{{\bf 0}\}$. If $\C$ is additive, we say $\C$ is {\em additive complementary dual (ACD) with respect to the trace Hermitian inner product} if $\C \cap \C^{\perp_{{\rm TrH}}}=\{{\bf 0}\}$, and $\C$ is {\em additive complementary dual (ACD) with respect to the trace Euclidean inner product} if $\C \cap \C^{\perp_{{\rm TrE}}}=\{{\bf 0}\}$.

\section{ACD codes with respect to the trace Hermitian inner product $\star$}
We have a criteria for ACD codes in terms of a generator matrix as follows.

\begin{theorem}{\em(\cite[Proposition 1]{Car}, \cite[Proposition 1]{Massey92})} \label{LG}
Let $G$ be a generator matrix for an $[n,k]$ linear code $\C$ over a field. Then $\C$ is a Euclidean (resp. a Hermitian) LCD
code if and only if, the $k \times k$ matrix $GG^T$ (resp. $G\bar{G}^T$) is invertible.
\end{theorem}

\begin{theorem}{\em(\cite{BooJit})} \label{thm-tH-ACD}
Let $\C$ be an additive $(n, 2^k)$ code over $\F_4$ with generator matrix $G$.
Then $\C$ is ACD with respect to the trace Hermitian inner product if and only if $G \star G=G \bar{G}^T + \bar{G}G^T$ is invertible.
\end{theorem}

The following lemma is straightforward from the definition of the inner products.

\begin{lem}{\em (\cite[Lemma 2.1]{DouKimLee})} \label{lem-dual-dual}
Assume $\C$ is a linear code over $\F_4$. Then
 $\C^{\perp_{{\rm TrH}}}$ is equal to
the dual of $\C$ with respect to the Hermitian inner product. Similarly,
$\C^{\perp_{{\rm TrE}}}$ is equal to the dual of $\C$ with respect to the Euclidean inner product.
\end{lem}

\begin{cor}\label{cor-herm-LCD}
Any Hermitian LCD $[n,k,d]$ code $\C$ over $\F_4$ is an ACD $(n,2^{2k}, d)$ code over $\F_4$ with respect to the trace Hermitian inner product.
\end{cor}

\begin{proof}
Suppose that $\C$ is a Hermitian LCD code over $\F_4$. Then $\C \cap \C^{\perp_{{\rm H}}} = \{ {\bf 0} \}$. By Lemma~\ref{lem-dual-dual},
$\C^{\perp_{{\rm H}}}=\C^{\perp_{{\rm TrH}}}$.
Thus, $\C \cap \C^{\perp_{{\rm TrH}}}=\{{\bf 0}\}$, which means that $\C$ is ACD with respect to the trace Hermitian inner product.
 The parameters of $\C$ over $\F_4$ are obvious.
\end{proof}

\begin{cor}
Suppose $\C_2$ is a binary LCD $[n, k, d]$ code with generator
matrix $G_2$. Let $\C_2^4$ be a linear code over $\F_4$ with generator matrix $G_2$. Then by regarding $\C_2^4$ as an additive code over $\F_4$, $\C_2^4$ is an ACD $(n, 2^{2k}, d)$ code with respect to the trace Hermitian inner product.
\end{cor}

\begin{proof}
Let $G_2$ be a generator matrix for $\C_2$. Then, by Theorem~\ref{LG}, $G_2 G_2^T$ is invertible over $\F_2$. Since $G_2 =\bar{G_2}$,
$G_2 \bar{G_2}^T$ is invertible over $\F_4$. Hence, by Theorem~\ref{LG}, $\C_2^4$ as a linear code over $\F_4$ is an Hermitian LCD $[n, k, d]$ code over $\F_4$. Therefore, by Corollary~\ref{cor-herm-LCD}, $\C_2^4$ is an ACD $(n, 2^{2k}, d)$ code with respect to the trace Hermitian inner product.
\end{proof}

%\noindent

\begin{lem}{\em(\cite{HPR})} \label{lem-2-rank}
If $A$ is a symmetric integral matrix with zero diagonal, then 2-rank$(A)$ is even.
\end{lem}
\begin{proof} We give a detailed proof here since the proof in~\cite{HPR} is concise.
Recall that a principal submatrix of a square matrix $A$ is the matrix obtained by deleting any $m$ rows and the corresponding $m$ columns.

Let $A'$ be a non-singular principal submatrix of $A$ such that 2-rank$(A)$ = 2-rank$(A')$. Then $A'$ is also symmetric integral with zero diagonal. We may assume that $A' \equiv B \pmod{2}$ for some skew symmetric integral matrix $B$.
Over $\mathbb Z$, any skew symmetric matrix $B$ of odd order has determinant $0$ (since $B=-B^T$ implies that $\det(B)=-\det(B^T)=-\det(B)$, hence $\det(B)=0$).
So, if 2-rank $(A')$ is odd, then $\det(A') \equiv \det(B) \equiv 0 \pmod{2}$. This is a contradiction since $A'$ has a full rank, that is, $\det(A') \not \equiv 0 \pmod{2}$. Thus 2-rank $(A')$ is even. Therefore, 2-rank $(A)$ is even.
\end{proof}

\begin{theorem}
If $\C$ is a trace Hermitian ACD $(n, 2^k)$ code over $\F_4$ with generator matrix $G$, then $k$ is even.
\end{theorem}

\begin{proof}
Let $A=G\star G=G\bar{G}^T+\bar{G}G^T$. Then $A$ is a $k\times k$ symmetric integral matrix with zero diagonal. Since $\C$ is ACD, the 2-rank of $A$ must be $k$ by Theorem~\ref{thm-tH-ACD}. By Lemma~\ref{lem-2-rank}, $k$ is even.
\end{proof}

\begin{lem} \label{lem-C4}
Let $\C$ and $\D$ are two binary $[n,k_1]$ and $[n,k_2]$ linear codes. If $\C_4=a\C + b\D$, where $a\neq b$ and $a,b \in \F_4\setminus \{0\}$, then $\C_4^{{\perp_{{\rm TrH}}}}=a\D^\perp + b\C^{\perp}$.
\end{lem}

\begin{proof}
Let $\textbf u \in a\D^\perp+b\C^{\perp}$; then there exist $\textbf d'\in \D^{\perp}$ and $\textbf c'\in \C^{\perp}$ such that $\textbf u=a\textbf d'+b\textbf c'$. For $\textbf v=a\textbf c+b\textbf d \in \C_4$, where $\textbf c\in \C$ and $\textbf d\in \D$, we have

\begin{eqnarray*}
 \textbf u\star \textbf v&=&(a\textbf d'+b\textbf c')\star (a\textbf c+b\textbf d)\\
&=&(a\textbf {d}'+b\textbf {c}')\cdot(\bar{a}\textbf c+\bar{b}\textbf d)+(\bar{a}\textbf {d}'+\bar{b}\textbf {c}')\cdot(a\textbf c+b\textbf d)\\
&=&a\bar{a}\textbf {d}'\cdot\textbf c+a\bar{b}\textbf {d}'\cdot\textbf d+b\bar{a}\textbf {c}'\cdot\textbf c+b\bar{b}\textbf {c}'\cdot\textbf d+\bar{a}a\textbf {d}'\cdot \textbf c+\bar{a}b\textbf d'\cdot\textbf d+\bar{b}a\textbf c'\cdot\textbf c+\bar{b}b\textbf {c}'\cdot\textbf d\\
&=&(a\bar{b}+\bar{a}b)\textbf {d}'\cdot\textbf d+(b\bar{a}+\bar{b}a)\textbf {c}'\cdot\textbf c\\
&=&0.
\end{eqnarray*}
Hence, $a\D^{\perp}+b\C^{\perp}\subseteq \C_4^{{\perp_{{\rm TrH}}}}$.

Since $a\neq b$, $|\C_4|=|\C| |\D|=2^{k_{1}}\cdot 2^{k_{2}}=2^{k_{1}+k_{2}}$, then $\C_4^{{\perp_{{\rm TrH}}}}=\frac{2^{2n}}{|\C_4|}=2^{2n-(k_{1}+k_{2})}=2^{2n-k_{1}-k_{2}}$. And $|a\D^{\perp}+b\C^{\perp}|=2^{n-k_{1}}\cdot 2^{n-k_{2}}=2^{2n-k_{1}-k_{2}}$. Therefore, $\C_4^{{\perp_{{\rm TrH}}}}=a\D^{\perp} + b\C^{\perp}$.

\end{proof}

\begin{proposition}\label{prop-ACD from binary self-dual codes}
If $\C$ is a self-dual $[2n,n]$ binary code, let $\D$ be a binary linear code and $\F_2^{2n}=\C\oplus \D$, then $\C_4=a\C + b\D$ is an ACD $(2n,2^{2n})$ code over $\F_4$ with respect to the trace Hermitian inner product, where $a\neq b$ and $a,b\in \F_4\setminus \{ 0\}$.
\end{proposition}

\begin{proof}
Since $\C$ is self-dual, $\C=\C^{\perp}$. By $\F_2^{2n}=\C\oplus \D$, we know that
$$\C\cap \D=\{\textbf 0\},\ \C^{\perp}\cap \D^{\perp}=(\C\oplus \D)^{\perp}=(\F_2^{2n})^{\perp}=\{\textbf 0\},$$
then $\C\cap \D^{\perp}=\{\textbf 0\}$.
By Lemma~\ref{lem-C4}, $\C_4^{\perp_{{\rm TrH}}}=a\D^{\perp}+b\C^{\perp}=a\D^{\perp}+b\C$.

If there exists ${\textbf 0}\neq {\textbf v} \in \C_{4}\cap \C_{4}^{\perp_{{\rm TrH}}}$,
then there are ${\textbf c}\in \C, {\textbf d}\in \D$, and ${\textbf c}, {\textbf d}$ are nonzero such that ${\textbf v}=a {\textbf c}+ b {\textbf d}$.
Similarly, there are ${\textbf c}'\in \D^{\perp}$, ${\textbf d}'\in \C$,
and ${\textbf c}', {\textbf d}'$are nonzero such that${\textbf v}=a {\textbf c}'+b {\textbf d}'$.
Hence we have $a {\textbf c}+b {\textbf d}=a {\textbf c}'+b {\textbf d}'$, which implies that ${\textbf c}={\textbf c}',{\textbf d}={\textbf d}'$. Then we have ${\textbf c}\in \C \cap \D^{\perp}$, ${\textbf d}\in \D \cap \C$, where ${\textbf c}, {\textbf d}$ are nonzero, which is a contradiction. Therefore, $\C_{4}$ is an ACD code.
\end{proof}

\begin{cor} Suppose that $\C$ is a self-dual $[2n,n]$ binary code and $\D$ is a binary linear code such that $\F_2^{2n}=\C\oplus \D$.
Let $\C_4=a\C + b\D$. Then its
minimum distance $d(\C_4)$ is  equal to $\min\{d(\C), d(\D)\}$.
\end{cor}
\begin{proof}
Obviously, we have $d(\C_4)\leq \min\{d(\C),d(\D)\}$. On the other hand, for any nonzero codeword $a\textbf u+b\textbf v \in \C_4$ with $\textbf u \in \C$ and $\textbf v\in \D$, since $a,b\in \F_4\setminus \{0\}$ and $a \neq b$, ${\rm wt}(a\textbf u+b\textbf v)\geq \max \{{\rm wt}(a\textbf u),{\rm wt}(b\textbf v)\} \geq \min\{d(\C),d(\D)\}$. Therefore, $d(\C_4)=\min\{d(\C), d(\D)\}$.
\end{proof}
%\medskip

%\noindent
%{\bf Conjecture} $\C_4=a\C+b\D = \{d(\C), d(\D)\}$.

%\medskip

\begin{remark}
Proposition~\ref{prop-ACD from binary self-dual codes} shows that we can get ACD codes over $\F_4$ from binary linear self-dual codes.
\end{remark}

\begin{example}
Let $\C$ be a binary $[4,2,2]$ code with generator matrix $G_1$ of the form
$$
G_1=\begin{bmatrix}
1 & 1 & 0 & 0 \\
0 & 0 & 1 & 1
\end{bmatrix}.
$$
Let $\D$ be a binary $[4,2,2]$ code with generator matrix $G_2$ of the form
$$
G_2=\begin{bmatrix}
1 & 1 & 0 & 1 \\
0 & 1 & 1 & 1
\end{bmatrix}.
$$
It is easy to check that $\C$ is self-dual and $\F_2^4=\C\oplus \D$. By Proposition~\ref{prop-ACD from binary self-dual codes}, $\C_4=\C+\omega \D$ is an $(4,2^4,2)$ ACD code over $\F_4$ with generator matrix $G$ of the form
$$
G=\begin{bmatrix}
1 & 1 & 0 & 0 \\
0 & 0 & 1 & 1 \\
\omega & \omega & 0 & \omega \\
0 & \omega & \omega & \omega
\end{bmatrix}.
$$
\end{example}

%{\bf Question: In the second line of the above proof
%why is $\C^{\perp}\cap \D^{\perp}=(\C\oplus \D)^{\perp}$?}

%{\bf Question: Is there a good example for Prop. 1? How do we choose $\D$ so that min distance of $\C_4$ is good?}

\medskip

\section{ACD codes with respect to the trace Euclidean inner product $\diamond$}

In this section, we construct ACD codes with respect to the trace Euclidean inner product.

\begin{lem}\label{cor-euc-LCD}
Any Euclidean LCD $[n,k,d]$ code $\C$ over $\F_4$ is an ACD $(n,2^{2k}, d)$ code over $\F_4$ with respect to the trace Euclidean inner product.
\end{lem}

\begin{proof}
Suppose that $\C$ is an Euclidean LCD code over $\F_4$. Then $\C \cap \C^{\perp_{{\rm E}}} = \{ {\bf 0} \}$, where $\C^{\perp_{{\rm E}}}$ denotes the dual of $\C$ under the Euclidean inner product. By Lemma~\ref{lem-dual-dual},
$\C^{\perp_{{\rm E}}}=\C^{\perp_{{\rm TrE}}}$.
Thus, $\C \cap \C^{\perp_{{\rm TrE}}}=\{{\bf 0}\}$, which means that $\C$ is ACD with respect to the trace Euclidean inner product.
 The parameters of $\C$ over $\F_4$ are obvious.
\end{proof}

We want a characterization of an ACD code with respect to the Euclidean inner product in terms of its generator matrix. We follow the idea from~\cite{BooJit}.

\begin{definition}{(\cite{BooJit})}\label{LP}
Let $V$ be an inner product space over a field $\F_q$. An $\mathbb{F}_q$-linear map $T:V\rightarrow V$ is called an $\mathbb{F}_q$-orthogonal projection with respect to the prescribed inner product $< \cdot,\cdot >$ if \\
 (i) $T^2=T$, and \\
 (ii) $<\textbf u,\textbf v>=0$ for all $\textbf u \in {\rm Im}(T)$ and $\textbf v\in {\rm Ker}(T)$.

\end{definition}

\begin{lem}\label{UD}
Using the notation of Definition 2, if $\langle \cdot,\cdot \rangle$ is nondegenerate, then ${\rm Ker}(T)$ is the dual of ${\rm Im}(T)$ under $\langle \cdot,\cdot \rangle$.
\end{lem}
\begin{proof}
From basic linear algebra, ${\rm dim}_{\F_q}({\rm Ker}(T))={\rm dim}_{\F_q}(V)-{\rm dim}_{\F_q}({\rm Im}(T))$. Let $\D$ be the dual of ${\rm Im}(T)$ under $\langle \cdot,\cdot \rangle$.  As $\langle \cdot,\cdot \rangle$ is nondegenerate, ${\rm dim}_{\F_q}(\D)={\rm dim}_{\F_q}(V)-{\rm dim}_{\F_q}({\rm Im}(T))$. Thus ${\rm Ker}(T)$ and $\D$ have the same dimension. By part (ii) of Definition 2, ${\rm Ker}(T)\subseteq \D$. As ${\rm Ker}(T)$ and $\D$ have the same dimension, they are equal.
\end{proof}

\begin{lem}{\label{lemma-orthogonal projection}}

 Let $\C$ be a linear code of length $n$ over $\F_4$ and let $T:\F_4^n\rightarrow \F_4^n$ be an $\F_4$-linear map. Then $T$ is an $\F_4$-orthogonal projection with respect to the trace Euclidean inner product onto $\C$ if and only if
\begin{eqnarray*}
T(\textbf v)=
\begin{cases}
 \textbf v & if \ \textbf v\in \C \\
 \textbf 0 &  if \ \textbf v\in {\C}^{\perp_{{\rm TrE}}}.
 \end{cases}
\end{eqnarray*}
\end{lem}

\begin{proof}
Suppose that $T:\F_4^n\rightarrow \F_4^n$ is an $\F_4-$orthogonal projection with respect to the trace Euclidean inner product onto $\C$.  Let ${\bf v}\in \C={\rm Im}(T)$. Then there exists ${\bf x}\in \F_4^n$ such that ${\bf v}=T({\bf x})$. So ${\bf v}=T({\bf x})=T^2({\bf x})=T(T({\bf x}))=T({\bf v})$. Now let ${\bf v}\in \C^{\perp_{{\rm TrE}}}$; then $T({\bf v})={\bf 0}$ by Lemma~\ref{UD}.

Conversely, assume that
\begin{eqnarray*}
T(\textbf v)=
\begin{cases}
 \textbf v & if \ \textbf v\in \C \\
 \textbf 0 &  if \ \textbf v\in {\C}^{\perp_{{\rm TrE}}}.
 \end{cases}
\end{eqnarray*}
Since $T$ is a function, $\C \cap\C^{\perp_{{\rm TrE}}}=\{{\bf 0}\}$ implying $\F_4^n=\C\oplus {\C}^{\perp_{{\rm TrE}}}$. If ${\bf v}\in \C$, $T^2({\bf v})=T(T({\bf v}))=T({\bf v})={\bf v}$; if ${\bf v}\in {\C}^{\perp_{{\rm TrE}}}$, $T^2({\bf v})=T(T({\bf v}))=T({\bf 0})={\bf 0}=T({\bf v})$. So $T^2=T$ on $\C$ and on ${\C}^{\perp_{{\rm TrE}}}$, and hence on $\C\oplus {\C}^{\perp_{{\rm TrE}}}=\F_4^n$ by linearity, verifying part (i) of Definition 2.  Also ${\rm Im}(T)=\C$.  As in the proof of Lemma~\ref{UD}, ${\rm dim}_{\F_4}(\C^{\perp_{{\rm TrE}}})={\rm dim}_{\F_4}({\rm Ker}(T))$. As $T({\bf v})={\bf 0}$ for ${\bf v}\in \C^{\perp_{{\rm TrE}}}$, $\C^{\perp_{{\rm TrE}}} \subseteq{\rm Ker}(T)$ imply $\C^{\perp_{{\rm TrE}}}={\rm Ker}(T)$, verifying part (ii) of Definition~\ref{LP}.
\end{proof}

\begin{lem}\label{lemma-orthogonal projection iff ACD}
Let $\C$ be a linear code of length $n$ over $\F_4$. Then $\C$ is ACD with respect to the trace Euclidean inner product if and only if there exists an $\F_4$-orthogonal projection with respect to the trace Euclidean inner product from $\F_4^n$ onto $\C$.
\end{lem}

\begin{proof}
Let $T_{\small{\C}}$ is an $\F_4$-orthogonal projection with respect to the trace Euclidean inner product from $\F_4^n$ onto $\C$. By Lemma~\ref{lemma-orthogonal projection}, it follows that,
\begin{eqnarray*}
T_{\small{\C}}({\textbf v})=
\begin{cases}
 \textbf v & if \ \textbf v\in \C \\
 \textbf 0 &  if \ \textbf v\in {\C}^{\perp_{{\rm TrE}}}.
 \end{cases}
\end{eqnarray*}

Assume that $\C$ is not ACD with respect to the trace Euclidean inner product. Then there exists $\textbf u\neq \textbf 0$ such that $\textbf u\in \C\cap \C^{\perp_{{\rm TrE}}}$. Hence, $ \textbf u=T_{\small{\C}}(\textbf u)=\textbf 0$, which is a contradiction. Therefore, $\C$ is ACD with respect to the trace Euclidean inner product

Conversely, assume that $\C$ is ACD with respect to the trace Euclidean inner product. Let $\textbf v\in \F_4^n$, then there exists a unique pair $\textbf u\in \C$ and $\textbf w\in \C^{\perp_{{\rm TrE}}}$ such that $\textbf v=\textbf u+\textbf w$.
Defined a map $T_{\small{\C}}:\F_4^n\rightarrow \F_4^n$ by $T_{\small{\C}}(\textbf v)=\textbf u$. Clearly, $T_{\small{\C}}$ is an $\F_4$-linear map such that:
\begin{eqnarray*}
T_{\small{\C}}({\textbf v})=
\begin{cases}
 \textbf v & if \ \textbf v\in \C \\
 \textbf 0 &  if \ \textbf v\in {\C}^{\perp_{{\rm TrE}}}.
 \end{cases}
\end{eqnarray*}
Hence, by Lemma~\ref{lemma-orthogonal projection}, $T_{\small{\C}}$ is an $\F_4$-orthogonal projection with respect to the trace Euclidean inner product from $\F_4^n$ onto $\C$.
\end{proof}

\begin{theorem}\label{LAG}
Let $\C$ be an additive $(n, 2^k)$ code over $\F_4$ with generator matrix $G$.
Then $\C$ is ACD with respect to the trace Euclidean inner product if and only if $G\diamond G=G G^T +\bar{G}\bar{G}^{T}$ is invertible. Moreover, in this case the map $T_{\small{\C}}({\textbf v})={\rm Tr}(\textbf vG^T)(GG^T+\bar{G}\bar{G}^{T})^{-1}G$ is an $\F_4$-orthogonal projection with respect to the trace Euclidean inner product
from $\F_4^n$ onto $\C$, where for ${\bf v}\in \F_4^n$, ${\rm Tr}({\bf v}G^T)={\bf v}G^T+\bar{{\bf v}}\bar{G}^T$.
\end{theorem}

\begin{proof}
Assume ${\rm Tr}(GG^T)=GG^T+\bar{G}\bar{G}^{T}$ is not invertible. Since ${\rm Tr}(GG^{T})$ is a $k\times k$ matrix, we have ${\rm rank}({\rm Tr}(GG^{T})< k$. Hence
$$k={\rm null}({\rm Tr}(GG^{T}))+{\rm rank}({\rm Tr}(GG^{T}))<{\rm null}({\rm Tr}(GG^{T}))+k,$$
 then ${\rm null}({\rm Tr}(GG^{T}))>k-k=0$.
So there exists $\textbf u\in {{\rm Ker}({\rm Tr}(GG^{T}))\setminus \{\textbf 0\}} \subseteq \F_2^k$ such that $\textbf u{\rm Tr}(GG^{T})={\textbf 0}$ and ${\textbf u}G\in {\C\setminus \{\textbf 0\}}$. We have,
$${\bf 0}\neq {\bf u}\in {\rm Ker}({\rm Tr}(GG^T))={\bf u}GG^T+{\bf u}\bar{G}\bar{G}^T=({\bf u}G)G^T+(\bar{{\bf u}}\bar{G})\bar{G}^T$$
Hence, $\textbf uG$ is also a vector in $\C^{\perp_{{\rm TrE}}}$; i.e., $\C\cap \C^{\perp_{{\rm TrE}}}\neq \{\textbf 0\}$. Therefore, $\C$ is not ACD with respect to the trace Euclidean inner product.

 Conversely, assume that $G G^T +\bar{G}\bar{G}^{T}$ is invertible. Let $T_{\small{\C}}:\F_4^n\rightarrow \C $ be defined by
 $$T_{\small{\C}}({\textbf v})={\rm Tr}({\textbf v}G^T)(GG^T+\bar{G}\bar{G}^{T})^{-1}G.$$
  Let ${\textbf v}\in \F_4^n$. If $\textbf v\in \C$, then there exists $\textbf u\in \F_2^k$ such that $\textbf v=\textbf uG$; hence,
\begin{eqnarray*}
T_{\small{\C}}({\textbf v})&=&{\rm Tr}(\textbf vG^T)(GG^T+\bar{G}\bar{G}^{T})^{-1}G\\
 &=&{\rm Tr}(\textbf uGG^T)(GG^T+\bar{G}\bar{G}^{T})^{-1}G\\
 &=&(\textbf uGG^{T}+\bar{\textbf u}\bar{G}\bar{G}^{T})(GG^T+\bar{G}\bar{G}^{T})^{-1}G\\
 &=&{\textbf u}(GG^{T}+\bar{G}\bar{G}^{T})(GG^T+\bar{G}\bar{G}^{T})^{-1}G\\
 &=&{\textbf u}I_kG\\
 &=&{\textbf u}G\\
 &=&{\textbf v},
\end{eqnarray*}
Assume that $\textbf v\in \C^{\perp_{{\rm TrE}}}$. Then ${\rm Tr}(\textbf vG^{T})={\textbf 0}$, and $$T_{\small{\C}}(\textbf v)={\rm Tr}(\textbf vG^T)(GG+\bar{G}\bar{G}^{T})^{-1}G={\textbf 0}(GG^T+\bar{G}\bar{G}^{T})^{-1}G={\textbf 0}.$$
Hence by Lemma~\ref{lemma-orthogonal projection}, $T_{\small{\C}}$ is an $\F_4$-orthogonal projection with respect to the trace Euclidean inner product from $\F_4^n$ onto $\C$. By Lemma~\ref{lemma-orthogonal projection iff ACD}, $\C$ is an ACD code with respect to the trace Euclidean inner product.
\end{proof}

\begin{cor}
If $\C$ and $\D$ are two binary LCD $[n, k_1]$ and $[n, k_2]$ codes respectively, then $\C_4=\omega \C + \omega^{2}\D$ is an ACD $(n, 2^{k_1 + k_2})$ code over $\F_4$ with respect to the trace Euclidean inner product.
\end{cor}

\begin{proof}
Since $\C$ and $\D$ are LCD codes, we have $\C\cap \C^{\perp}=\{\textbf 0\}, \D \cap \D^{\perp}=\{\textbf 0\}$. Furthermore, since $\{ \omega,\omega^{2}\}$ is a trace orthogonal basis in $\F_4$, we have $\C_{4}^{\perp_{{\rm TrE}}}=\omega \C^{\perp}+\omega^{2} \D^{\perp}$.
We need to prove that $\C_{4} \cap \C_{4}^{\perp_{{\rm TrE}}}=\{\textbf 0\}$. If there exists $\textbf 0\neq \textbf v \in \C_{4}\cap \C_{4}^{\perp_{{\rm TrE}}}$, then there are $\textbf c\in \C, \textbf d\in \D$, and $\textbf c, \textbf d$ are nonzero such that $\textbf v=\omega \textbf c+\omega^{2} \textbf d$.
Similarly, there are $\textbf c'\in \C^{\perp}$, $\textbf d'\in \D^{\perp}$, and $\textbf c', \textbf d'$ are nonzero such that $\textbf v=\omega \textbf c'+\omega^{2} \textbf d'$.
Hence we have $\omega \textbf c+\omega^{2} \textbf d=\omega \textbf c'+\omega^{2} \textbf d'$, which implies that $\textbf c=\textbf c',\textbf d=\textbf d'$. Then we have $\textbf c\in \C \cap \C^{\perp}$, $\textbf d\in \D \cap \D^{\perp}$, where $\textbf c, \textbf d$ are nonzero, which is a contradiction. Therefore, $\C_{4}$ is an ACD code. Clearly, the $\F_2$-rank of $\C_4$ is $k_1 + k_2$.
\end{proof}

\begin{cor}\label{cor-acd codes from LCD}
Let $\C_1$ be a binary $[2n,k]$  code with generator matrix $G_1=[A|B]$, where $A$ and $B$ are $k\times n$ matrices. Let $\C_2$ be an additive code with generator matrix $G=\omega A+\omega^2 B$. Then $\C_2$ is an ACD $(n, 2^{k})$ code over $\F_4$ with respect to the trace Euclidean inner product if and only if $\C_1$ is a binary $[2n,k]$ LCD code.
\end{cor}

\begin{proof}
Suppose that $\C_1$ is a binary LCD code. Then, by Theorem~\ref{LG}, $[A|B][A|B]^{T}=AA^{T}+BB^{T}$ is invertible over $\F_2$. We also have
\begin{eqnarray*}
G\diamond G&=&(\omega A+\omega^{2} B)(\omega A^{T}+\omega^{2}B^{T})+(\overline{\omega} A+\overline{\omega^{2}} B)(\overline{\omega} A^{T}+\overline{\omega^{2}}B^{T})\\
&=&(\omega^{2}+\omega)AA^{T}+(\omega +\omega^{2})BB^{T}\\
&=&AA^{T}+BB^{T};
\end{eqnarray*}
hence $G\diamond G=[A|B][A|B]^{T}$ is invertible over $\F_4$. Therefore, by Theorem~\ref{LAG}, $\C_2$ is ACD with respect to the trace Euclidean inner product. Clearly, the $\F_2$-rank of $\C_2$ is $k$.

Conversely, if $\C_2$ is ACD with respect to the trace Euclidean inner product, then we reverse the above proof to show that $\C_1$ is a binary LCD code.
\end{proof}

\begin{example}
Let $\C_1$ be a binary $[12,6,4]$ code with the generator matrix $G_1=[A|B]$ of the form

$$G_1=\left[
\begin{array}{c|c}
\begin{matrix}
1 & 0 & 0 & 0 & 0 & 0  \\
0 & 1 & 0 & 0 & 0 & 0 \\
0 & 0 & 1 & 0 & 0 & 0  \\
0 & 0 & 0 & 1 & 0 & 0  \\
0 & 0 & 0 & 0 & 1 & 0  \\
0 & 0 & 0 & 0 & 0 & 1
\end{matrix}&
\begin{matrix}
0 & 1 & 1 & 1 & 0 & 0 \\
1 & 1 & 0 & 1 & 1 & 0 \\
1 & 0 & 1 & 0 & 1 & 1 \\
1 & 1 & 0 & 1 & 0 & 1 \\
0 & 1 & 1 & 0 & 1 & 1 \\
0 & 0 & 1 & 1 & 1 & 0
\end{matrix}
\end{array}
\right].
$$
It is easy to check that $\C_1$ is a binary LCD code. By Corollary~\ref{cor-acd codes from LCD} $\C_2$ is a $(6,2^6,4)$ ACD code with generator matrix $G$
$$
G=\begin{bmatrix}
\omega & \omega^2 & \omega^2 & \omega^2 & 0 & 0 \\
\omega^2 & 1 & 0 & \omega^2 & \omega^2 & 0 \\
\omega^2 & 0 & 1 & 0 & \omega^2 & \omega^2 \\
\omega^2 & \omega^2 & 0 & 1 & 0 & \omega^2 \\
0 & \omega^2 & \omega^2 & 0 & 1 & \omega^2 \\
0 & 0 & \omega^2 & \omega^2 & \omega^2 & \omega
\end{bmatrix}.
$$
\end{example}

\begin{cor}
Let $\C$ be an $(n, 2^k)$ additive conjucyclic code with generator matrix $G$ over $\F_4$, and form the binary code

$$\C' = \{ \Tr(\om \textbf u) | \Tr(\ob \textbf u) : \textbf u \in \C \},$$

\noindent
where the trace is applied componentwise and the vertical bar denotes concatenation. Then $\C'$ is a
binary cyclic code of length $2n$ with generator matrix $G'=[\omega G+\bar{\omega}\bar{G}|\bar{\omega}G+\omega \bar{G}]$, which is LCD if and only if $\C$ is ACD with respect to the trace Euclidean inner product.
\end{cor}
\begin{proof}
Suppose that $\C$ is an ACD code with respect to the trace Euclidean inner product. Then $G\diamond G=GG^T+\bar{G}\bar{G^T}$ is invertible. And $\C'$ is a
binary cyclic code with generator matrix $G'$. Hence:
\begin{eqnarray*}
G'G'^T&=&(\omega G+\bar{\omega}\bar{G})(\omega G^T+\bar{\omega}\bar{G}^T)+(\bar{\omega} G+\omega\bar{G})(\bar{\omega} G^T+\omega\bar{G}^T)\\
&=&\omega ^2GG^T+G\bar{G}^T+\bar{G}G^T+\omega \bar{G}\bar{G}^T+\omega GG^T+G\bar{G}^T+\bar{G}G^T+\omega ^2\bar{G}\bar{G}^T\\
&=&GG^T+\bar{G}\bar{G}^T.
\end{eqnarray*}
Therefore, $\C'$ is a binary LCD code.

Conversely, if  $\C'$ is a binary LCD code, then we reverse the above proof to show that $\C$ is an ACD code with respect to the trace Euclidean inner product.
\end{proof}

\begin{remark}
 As a natural question, one can ask whether there is
 an ACD $(n, 2^{k^*}, d^*)$ code over $\mathbb F_4$ under the trace Euclidean inner product which satisfies $k^* > 2k$ and $d^* = d$, given an optimal linear LCD $[n, k, d]$ code over $\mathbb F_4$ under the Euclidean inner product. In what follows, we give several examples with the above conditions. This implies that ACD codes over $\mathbb F_4$ are sometimes better than LCD codes over $\mathbb F_4$.
 \end{remark}
   % delete later
\begin{example}
By Grassl's table~\cite{Gra}, there is an Euclidean optimal $[6,2,4]$ code over $\F_4$. One can also find an optimal  LCD $[6 ,2, 4]$ code $\mathcal K_1$ over $\F_4$ with generator matrix $K_1$ under the Euclidean inner product.
$$
K_1=\begin{bmatrix}
1 & 0 & 1 & 1 & 0 & \omega \\
0 & 1 & \omega & \omega & \omega^2 & \omega
\end{bmatrix},
$$
By a random search, we have constructed an ACD $(6,2^5,4)$ code $\mathcal K_{2,1}$ over $\F_4$ with generator matrix $K_{2,1}$ under the trace Euclidean inner product as follows. Note that $\mathcal K_{2,1}$ has double codewords than $\mathcal K_1$ although both have the same length and minimum weight.
$$
K_{2,1}=\begin{bmatrix}
1 & 0 & \omega & 1 & \omega^2 & \omega \\
\omega & 0 & \omega & \omega^2 & 1 & \omega^2 \\
0 & 1 & 0 & 1 & \omega & \omega \\
0 & \omega & \omega & \omega & \omega & \omega \\
0 & 0 & 1 & \omega & 1 & \omega
\end{bmatrix}.$$ % delete later
The weight distribution of $\mathcal K_{2,1}$ is $A_0=1, A_4=17, A_5=8, A_6=6$ and the order of the permutation automorphism group of $\mathcal K_{2,1}$ is 4. \\
\indent We have also found two more inequivalent ACD $(6,2^5,4)$ codes denoted by $\mathcal K_{2,2}$ and $\mathcal K_{2,3}$ with generator matrices $K_{2,2}$ and $K_{2,3}$, respectively.
The weight distributions and the orders of the permutation automorphism groups of these codes are displayed in Table 1.
$$
K_{2,2}=\begin{bmatrix}
  1  & 0  &  \omega  & \omega^2 & \omega^2 &   0\\
 \omega  &   0 &  0 & \omega &   1 &\omega^2 \\
  0  & 1 &  0  & 1 &  1  & \omega \\
  0  &\omega  & 0 & \omega  & \omega  &  1\\
  0  & 0 & \omega^2 & \omega^2  &  1  & 1
\end{bmatrix}, ~~~
K_{2,3}=\begin{bmatrix}
  1  & 0  & 0  & 1 &  \omega & \omega^2 \\
  \omega &  0  &  \omega & \omega^2 & \omega^2 & \omega^2 \\
  0  & 1 &  0  & 1 & \omega^2  & 1 \\
  0  & \omega &  \omega &  1  &  \omega &  0 \\
  0  & 0 & \omega^2 &  \omega & \omega^2 &  \omega\\
\end{bmatrix}.
$$
\end{example}

\begin{example}
It is well known~\cite{Car} that linear $[n,k, d]$ codes over $\F_4$ are equivalent to LCD $[n,k, d]$ codes under the Euclidean inner product. For $n=35$ and $96$ with $k=3$, there are optimal linear $[35,3,26]$, $[96,3,72]$ codes over $\F_4$~\cite{Gra}, which implies that there are optimal LCD $[35,3,26]$, $[96,3,72]$ codes over $\F_4$. Guo el al.~\cite{GUO} constructed better additive $(35,2^7,26)$ and $(96,2^7,72)$ codes over $\F_4$ with generator matrices $G_{7,35}$ and $G_{7,96}$, respectively. However, these are not ACD codes under the trace Euclidean inner product. By multiplying some columns of $G_{7,35}$ and $G_{7,96}$ by nonzero elements of $\F_4$, we have constructed an ACD $(35,2^7,26)$ code $\mathcal K_3$ and an ACD $(96,2^7,72)$ code $\mathcal K_4$ over $\F_4$ with generator matrices $K_3$ and $K_4$, respectively under the trace Euclidean inner product.
The weight distributions and the orders of the permutation automorphism groups of these codes are displayed in Table 1.
{\small\[{\setlength{\arraycolsep}{0.1pt}
\renewcommand{\arraystretch}{.5}
K_3=\left[\begin{array}{ccccccccccccccccccccccccccccccccccc}
1&0&\omega&\omega^2&0&\omega&\omega^2&0&1&\omega^2&\omega^2&1&0&\omega&\omega^2&\omega&0&0&\omega&\omega^2&1&\omega&\omega^2&1&1&\omega^2&\omega&0&\omega^2&\omega&0
&0&\omega&\omega^2&1 \\
\omega&0&\omega&1&\omega^2&1&1&0&\omega&0&1&\omega^2&1&\omega^2&1&\omega^2&1&0&\omega&0&\omega&\omega^2&1&\omega^2&\omega&0&\omega&0&0&\omega&0&1&\omega^2
&1&\omega^2 \\
0&1&\omega&0&\omega&\omega^2&1&0&1&\omega&0&\omega^2&\omega&1&0&1&\omega&0&\omega^2&\omega&1&0&1&\omega&0&\omega^2&\omega&1&\omega^2&\omega&1&\omega^2&0&1
&\omega \\
0&\omega&0&\omega&\omega&\omega&\omega^2&0&\omega&0&1&\omega^2&1&\omega^2&0&\omega&0&1&\omega^2&1&\omega^2&0&\omega&0&1&\omega^2&1&\omega^2&\omega&0&\omega&\omega^2
&1&\omega^2&1 \\
0&0&1&1&1&\omega^2&\omega^2&0&0&1&\omega^2&\omega&\omega&\omega^2&0&0&1&\omega^2&\omega&\omega&\omega^2&0&0&1&\omega^2&\omega&\omega&\omega^2&1&1&0&\omega
&\omega^2&\omega^2&\omega \\
0&0&0&0&0&0&0&1&1&1&1&1&1&1&\omega&\omega&\omega&\omega&\omega&\omega&\omega&\omega^2&\omega^2&\omega^2&\omega^2&\omega^2&\omega^2&\omega^2&1&1&1&1&1
&1&1 \\
0&0&0&0&0&0&0&\omega&\omega&\omega&\omega&\omega&\omega&\omega&\omega^2&\omega^2&\omega^2&\omega^2&\omega^2&\omega^2&\omega^2&1&1&1&1&1&1&1&\omega&\omega&\omega&\omega&\omega
&\omega&\omega
\end{array}\right].}\]}
$$K_4=\begin{bmatrix}
A_1 & A_2 & A_3
\end{bmatrix},$$
where
{\small\[{\setlength{\arraycolsep}{0.34pt}
\renewcommand{\arraystretch}{.5}
A_1=\left[\begin{array}{cccccccccccccccccccccccccccccccc}
1&0&0&1&\omega^2&1&1&\omega&\omega&0&0&\omega^2&\omega^2&1&1&\omega&\omega&0&0&\omega^2&\omega^2&1&1&\omega&0&\omega^2&\omega^2&\omega&\omega&1&1&0\\
\omega&0&1&1&1&\omega&0&1&\omega^2&0&\omega&\omega^2&1&\omega&0&1&\omega^2&0&\omega&\omega^2&1&\omega&0&1&0&1&\omega&0&\omega^2&\omega&1&\omega^2\\
0&1&0&1&\omega&\omega^2&\omega&1&0&1&0&\omega^2&\omega&\omega^2&\omega&1&0&1&0&\omega^2&\omega&\omega^2&\omega&1&0&\omega&0&1&\omega^2&1&\omega^2&\omega\\
0&\omega&1&\omega&1&\omega^2&\omega^2&0&0&\omega&\omega&1&1&\omega^2&\omega^2&0&0&\omega&\omega&1&1&\omega^2&\omega^2&0&0&\omega^2&\omega^2&\omega&\omega&1&1&0\\
0&0&\omega&\omega^2&1&1&\omega&\omega^2&0&0&\omega^2&\omega&1&1&\omega&\omega^2&0&0&\omega^2&\omega&1&1&\omega&\omega^2&0&0&1&\omega^2&\omega&\omega&\omega^2&1\\
0&0&0&0&0&0&0&0&0&0&0&0&0&0&0&0&0&0&0&0&0&0&0&0&1&\omega&\omega&\omega&\omega&\omega&\omega&\omega\\
0&0&0&0&0&0&0&0&0&0&0&0&0&0&0&0&0&0&0&0&0&0&0&0&\omega&\omega^2&\omega^2&\omega^2&\omega^2&\omega^2&\omega^2&\omega^2
\end{array}\right],}\]}
{\small\[{\setlength{\arraycolsep}{0.15pt}
\renewcommand{\arraystretch}{.5}
A_2=\left[\begin{array}{cccccccccccccccccccccccccccccccc}
0&\omega^2&\omega^2&\omega&\omega&1&1&0&\omega&1&1&0&0&\omega^2&\omega^2&\omega&\omega&1&1&0&0&\omega^2&\omega^2&\omega&1&\omega&\omega&\omega^2&\omega^2&0&0&1\\
0&1&\omega&0&\omega^2&\omega&1&\omega^2&\omega^2&\omega&1&\omega^2&0&1&\omega&0&\omega^2&\omega&1&\omega^2&0&1&\omega&0&\omega&\omega^2&0&\omega&1&0&\omega^2&1\\
0&\omega&0&1&\omega^2&1&\omega^2&\omega&0&\omega&0&1&\omega^2&1&\omega^2&\omega&0&\omega&0&1&\omega^2&1&\omega^2&\omega&0&\omega&0&1&\omega^2&1&\omega^2&\omega\\
0&\omega^2&\omega^2&\omega&\omega&1&1&0&0&\omega^2&\omega^2&\omega&\omega&1&1&0&0&\omega^2&\omega^2&\omega&\omega&1&1&0&0&\omega^2&\omega^2&\omega&\omega&1&1&0\\
0&0&1&\omega^2&\omega&\omega&\omega^2&1&0&0&1&\omega^2&\omega&\omega&\omega^2&1&0&0&1&\omega^2&\omega&\omega&\omega^2&1&0&0&1&\omega^2&\omega&\omega&\omega^2&1\\
\omega&\omega&\omega&\omega&\omega&\omega&\omega&\omega&\omega^2&\omega^2&\omega^2&\omega^2&\omega^2&\omega^2&\omega^2&\omega^2&\omega^2&\omega^2&\omega^2&\omega^2
&\omega^2&\omega^2&\omega^2&\omega^2&1&1&1&1&1&1&1&1\\
\omega^2&\omega^2&\omega^2&\omega^2&\omega^2&\omega^2&\omega^2&\omega^2&1&1&1&1&1&1&1&1&1&1&1&1&1&1&1&1&\omega&\omega&\omega&\omega&\omega&\omega&\omega&\omega
\end{array}\right],}\]}
{\small\[{\setlength{\arraycolsep}{0.23pt}
\renewcommand{\arraystretch}{.5}
A_3=\left[\begin{array}{cccccccccccccccccccccccccccccccc}
1&\omega&\omega&\omega^2&\omega^2&0&0&1&1&\omega&\omega&\omega^2&\omega^2&0&0&1&\omega^2&0&0&1&1&\omega&\omega&\omega^2&\omega^2&0&0&1&1&\omega&\omega&\omega^2\\
\omega&\omega^2&0&\omega&1&0&\omega^2&1&\omega&\omega^2&0&\omega&1&0&\omega^2&1&\omega&\omega^2&0&\omega&1&0&\omega^2&1&\omega&\omega^2&0&\omega&1&0&\omega^2&1\\
0&\omega&0&1&\omega^2&1&\omega^2&\omega&0&\omega&0&1&\omega^2&1&\omega^2&\omega&\omega&0&\omega&\omega^2&1&\omega^2&1&0&\omega&0&\omega&\omega^2&1&\omega^2&1&0\\
0&\omega^2&\omega^2&\omega&\omega&1&1&0&0&\omega^2&\omega^2&\omega&\omega&1&1&0&0&\omega^2&\omega^2&\omega&\omega&1&1&0&0&\omega^2&\omega^2&\omega&\omega&1&1&0\\
0&0&1&\omega^2&\omega&\omega&\omega^2&1&0&0&1&\omega^2&\omega&\omega&\omega^2&1&1&1&0&\omega&\omega^2&\omega^2&\omega&0&1&1&0&\omega&\omega^2&\omega^2&\omega&0\\
1&1&1&1&1&1&1&1&1&1&1&1&1&1&1&1&\omega&\omega&\omega&\omega&\omega&\omega&\omega&\omega&\omega&\omega&\omega&\omega&\omega&\omega&\omega&\omega\\
\omega&\omega&\omega&\omega&\omega&\omega&\omega&\omega&\omega&\omega&\omega&\omega&\omega&\omega&\omega&\omega&\omega^2&\omega^2&\omega^2&\omega^2&\omega^2&
\omega^2&\omega^2&\omega^2&\omega^2&\omega^2&\omega^2&\omega^2&\omega^2&\omega^2&\omega^2&\omega^2
\end{array}\right].}\]}
\end{example}

%\textcolor{blue}{
\begin{table}[h]
\label{table:ACD codes}
\footnotesize{
\caption{Weight distribution and permutation automorphism group order $|{\mbox{PAut}}(\mathcal C)|$ for our ACD codes $\mathcal C$ over $\mathbb F_4$}}
\centering
%\caption{\footnotesize{The Generators and WeightDistributions of additive self-dual codes up to length 29}}
\begin{tabular}{c|c|c|c}
\hline
code $\mathcal C$    & parameters & weight distribution  & $|{\mbox{PAut}}(\mathcal C)|$ \\
\hline
$\mathcal K_{2,1}$ & $(6, 2^5, 4)$ & $A_0=1, A_4=17, A_5=8, A_6=6$ & 4\\
$\mathcal K_{2,2}$ & $(6, 2^5, 4)$ & $A_0=1, A_4=17, A_5=8, A_6=6$ & 1 \\
$\mathcal K_{2,3}$ & $(6, 2^5, 4)$ & $A_0=1, A_4=15, A_5=12, A_6=4$ & 1 \\
$\mathcal K_{3}$   & $(35, 2^7, 26)$ & $A_0=1, A_{26}=105, A_{28}=15, A_{30}=7$ & 84 \\
$\mathcal K_{4}$   & $(96, 2^7, 72)$ & $A_0 =1, A_{72}=118, A_{80}=9$ & $2^{44} \cdot 3^{17} \cdot 7$ \\
\hline
\end{tabular}
\end{table}
%}

\section{Conclusion}

In this paper, we have studied ACD codes over $\mathbb F_4$ with respect to the trace Hermitian inner product and the trace Euclidean inner product. Interesting constructions of ACD codes from binary codes are given with respect to the both
inner products.  As a good motivation of ACD codes, we have also constructed several ACD $(n, 2^{2k+1}, d)$ codes over $\F_4$ under the trace Euclidean inner product which are better than optimal Euclidean LCD $[n, k, d]$ codes over $\mathbb F_4$.

\section*{Acknowledgement}
We want to thank the referees for their careful reading and constructive comments. This paper has been greatly improved.

\end{document}